\documentclass[runningheads]{llncs}

\def\doctitle{Modular Analysis of Tree-Topology Models}
\usepackage[bookmarks=false,pdftex,
				english,
        colorlinks=true,
        pdftitle={\doctitle},
        pdfauthor={Micha\l{} Knapik and Wojciech Penczek
        	and Laure Petrucci and Teofil Sidoruk},
        pdfsubject={},
        citecolor=blue,
        linkcolor=blue,
        ]{hyperref}

\usepackage[utf8]{inputenc}
\usepackage[english]{babel}
\usepackage[table,dvipsnames]{xcolor}

\usepackage{amssymb}
\usepackage{amsmath}
\usepackage{stmaryrd}
\usepackage{bbm}
\usepackage{url}
\usepackage{nicefrac}
\usepackage{mathtools}
\usepackage{color}
\usepackage{tikz}
\usepackage{comment}
\usepackage{graphicx}
\usepackage{url}
\usepackage{datetime}
\usepackage[normalem]{ulem}
\usepackage{xspace}
\usepackage{multicol,multirow,colortbl,array,booktabs}

\usepackage{algorithm}
\usepackage{algorithmic}
\usepackage{wrapfig}
\usepackage{array}
\usepackage{multirow}
\usepackage[bottom]{footmisc}
\usepackage{paralist}
\DeclareMathAlphabet{\mathpzc}{OT1}{pzc}{m}{it}
\usepackage{stmaryrd}
\usepackage{proof}
\usepackage{soul}
\usepackage{eurosym}
\usepackage{acronym}

\usepackage{array,etoolbox}

\usepackage{pdflscape}

\usepackage{marginnote}
\usepackage[capitalise]{cleveref}
\usepackage{subcaption}

\preto\tabular{\setcounter{magicrownumbers}{0}}
\newcounter{magicrownumbers}

\newcommand{\marginX}{\marginnote{\huge{\quad\textbf{!}\quad}}}

\ifdefined \VersionWithComments
  \newcommand{\mk}[1]{\textcolor{magenta}{\marginX{}[\textbf{Michał}: #1]}}
  \newcommand{\lp}[1]{\textcolor{green!70!black}{\marginX{}[\textbf{Laure}: #1 ]}}
\else
  \newcommand{\mk}[1]{}
  \newcommand{\lp}[1]{}
\fi

\newcommand{\ie}{i.e.\xspace}

\def\ADT/{ADTree}
\def\AMAS/{AMAS}
\def\EAMAS/{EAMAS}
\def\IIS/{IIS}
\def\EIIS/{EIIS}
\def\MAS/{MAS}
\def\POR/{POR}

\newcommand{\receive}{\ensuremath{?}}
\def\resp/{resp.}
\newcommand{\send}{\ensuremath{!}}
\newcommand{\styleSync}[1]{\textcolor{blue}{\ensuremath{#1}}}

\newcommand{\imitator}{\MakeUppercase{imitator}}
\newcommand{\spin}{\textsc{spin}\xspace}
\newcommand{\uppaal}{\textsc{Uppaal}\xspace}

\makeatletter
\renewcommand{\paragraph}{\@startsection{paragraph}{4}{0pt}%
  {.8ex plus 0.2ex minus 0.2ex}%
  {-0.5em}%
  {\bfseries}}
\makeatother

\newcommand{\CTL}{\mathit{CTL}}

\newcommand{\lts}{\ensuremath{\mathcal{L}\mathcal{T}\mathcal{S}}}
\newcommand{\snt}{\ensuremath{\mathcal{S}\mathcal{T}}\xspace}
\newcommand{\model}{\mathcal{M}}
\newcommand{\run}{\rho}
\newcommand{\runs}{\mathit{Runs}}
\newcommand{\aprod}{||}
\newcommand{\states}{\mathcal{S}}
\newcommand{\loc}{\ensuremath{s}}
\newcommand{\rtloc}{\ensuremath{r}}
\newcommand{\lloc}{\ensuremath{s}}
\newcommand{\rloc}{\ensuremath{t}}
\newcommand{\state}{\mathit{s}}

\newcommand{\initst}{\mathit{s}^0}
\newcommand{\initstsq}{\initst_{sq}}

\newcommand{\transrelg}{\rightarrow}

\makeatletter
\newcommand{\raisemath}[1]{\mathpalette{\raisem@th{#1}}}
\newcommand{\raisem@th}[3]{\raisebox{#1}{$#2#3$}}
\makeatother

\newcommand{\transc}[1]{\xrightarrow{#1}}
\newcommand{\transcr}[1]{\xrightarrow{#1}_{sq}}
\newcommand{\transci}[1]{\xrightarrow{#1}_{i}}
\newcommand{\transcrt}[1]{\xrightarrow{#1}_{\exroot}}
\newcommand{\acts}{\mathit{Acts}}
\newcommand{\downacts}{\mathit{downacts}}
\newcommand{\upacts}{\mathit{upacts}}
\newcommand{\getacts}{\mathit{acts}}
\newcommand{\getstates}{\mathit{states}}
\newcommand{\snd}{\mathit{snd}}
\newcommand{\locacts}{\mathit{locacts}}
\newcommand{\act}{\mathit{act}}

\newcommand{\props}{{\mathcal{P}\mathcal{V}}}
\newcommand{\labeling}{\mathcal{L}}

\newcommand{\syncn}{\mathcal{G}}
\newcommand{\sq}{\mathcal{S}\mathcal{Q}}
\newcommand{\cmpl}{\mathit{cmpl}}
\newcommand{\synce}{\mathcal{T}}

\newcommand{\nats}{\mathbb{N}}

\newcommand{\scomp}{\mathit{mods}}

\newcommand{\roottr}{\mathit{root}}
\newcommand{\prnt}{\mathit{parent}}
\newcommand{\chld}{\mathit{children}}

\newcommand{\exroot}{\mathcal{R}}
\newcommand{\excomp}{\mathcal{M}}
\newcommand{\exleft}{\excomp_1}
\newcommand{\exright}{\excomp_2}

\newcommand{\exrootn}{\mathcal{R}^y}

\newcommand{\exleftn}{\excomp^y_1}
\newcommand{\exrightn}{\excomp^y_2}

\newcommand{\actopen}{\mathit{open}}
\newcommand{\actidle}{\mathit{beep}}
\newcommand{\actsil}{\tau}
\newcommand{\actchooseL}{\mathit{chooseL}}
\newcommand{\actchooseR}{\mathit{chooseR}}

\newcommand{\projdown}[1]{{#1}\!\downarrow}

\newcommand{\reduce}{\ensuremath{\mathit{reduceNet}}}
\newcommand{\reducedchldn}{\mathit{redChldn}}

\newcommand{\chldalg}{\mathit{chld}}

\newcommand{\tnet}{\mathit{Net}}

\colorlet{leftCol}{orange!70}
\colorlet{rightCol}{purple!70}
\colorlet{rootCol}{green!70!black}

\excludecomment{CRUFT}
\excludecomment{CRUFTDEFS}
 
\usetikzlibrary{decorations,fit,arrows,automata,positioning,patterns,shapes,shapes.multipart,calc,matrix,backgrounds,shapes.gates.logic,shapes.gates.logic.US,circuits}
\tikzstyle{every node}=[initial text=]
\tikzstyle{final}=[double]
\tikzstyle{sync}=[draw=blue,thick]
\tikzstyle{nosync}=[draw=black,thick]
\tikzstyle{outersync}=[draw=purple,thick]
\tikzstyle{seq}=[path picture={
	\draw[->]  ([xshift=#1, yshift=6pt]path picture bounding box.south)
					-- ([xshift=#1, yshift=-7pt]path picture bounding box.south);}]

\newcount\xrownum
\makeatletter
 \def\@rowc@lors{\noalign{%
  \ifodd\xrownum
  \global\advance\rownum\@ne
  \fi
  \global\advance\xrownum\@ne
  }\@rowcolors}
\makeatother

\addto\captionsenglish{%
\floatname{algorithm}{Alg.}
}

\newcolumntype{x}[1]{>{\raggedleft\hspace{0pt}}p{#1}}

\title{%
Modular Analysis of Tree-Topology Models
\thanks{This work was supported by the PICS CNRS/PAN project PARTIES and by PAN.
  Michał Knapik is supported by POLLUX VoteVerif project.} }

\author{
	Laure Petrucci\inst{1}
	\and
        Micha\l{} Knapik\inst{2}
}
\institute{LIPN, CNRS UMR 7030, Universite Sorbonne Paris Nord,
                99 av. J-B. Clément, 93430
		Villetaneuse, France
                \email{laure.petrucci@lipn.univ-paris13.fr}
	\and
        {Institute of Computer Science,
        Polish Academy of Sciences,
                Jana Kazimierza 5, 01-248
		Warsaw, Poland
                \email{m.knapik@ipipan.waw.pl}
        }
}

\usepackage{cuted}
\begin{document}

\maketitle

\thispagestyle{plain}
\pagestyle{plain}

\begin{abstract}
We investigate networks of automata that synchronise over common
action labels. A graph synchronisation topology between the automata
is defined in such a way that two automata are connected iff
they can synchronise over an action. We show a very effective
reduction of networks of automata with tree-like synchronisation
topologies. The reduction preserves a certain form of reachability,
but not safety. The procedure is implemented in an open-source
tool.
\end{abstract}

\section{Introduction}
\label{sec:introduction}

Networks of various flavours of finite automata are the usual choice of formalism
when modeling complex systems such as protocols. This approach also plays well with the
divide-and-conquer paradigm, as the investigated system can be divided into components
modeled with various degree of granularity. However, the cost of computing the synchronised
product of these submodules can be prohibitive: in practice the size of the statespace
grows exponentially with the number of components.

In this paper we tackle the problem of computing of a part of the statespace
of the entire synchronised product in such a way that a certain version of reachability
is preserved. At this stage we only deal with systems that exhibit tree-like
synchronisation structure and consist of live-reset automata. Namely, each component
can synchronise via shared \emph{upstream} actions with a single other module (its \emph{parent})
after which it \emph{resets}, \ie{} returns to the initial state.
We propose a bottom-up reduction based on the observation that any execution of the
entire system can be rewritten in a reachability-preserving way into a sequence of interactions between
components and their parents followed by upstream synchronisations. Thus, the reduced model is
constructed by creating synchronised products of pairs consisting of a component and its parent.
The size of the statespace of the resulting automaton is much smaller than the product
of the entire network.

The theory has been implemented in an open-source tool~\cite{LTR}.

\section{Tree Synchronisation Systems}
\label{sec:tss}

In this section we recall the basic notions of networks of Labelled Transition Systems and their
synchronisation topologies. We also introduce and explain the restrictions on the models assumed in this
paper.
In what follows let $\props$ denote the set of propositions.

\begin{definition}[Labelled Transition System]
A \emph{Labelled Transition System} ($\lts$) is a tuple
$\model = \langle \states, \initst, \acts, \transrelg, \labeling\rangle$ where:
\begin{enumerate}
\item $\states$ is a finite set of states and $\initst\in\states$ the initial state;
\item $\acts$ is a finite set of action names;
\item $\transrelg\; \subseteq \states\times\acts\times\times \states$ is a transition relation;
\item $\labeling\colon\states\to 2^\props$ assigns to each state a set of propositions
that hold therein.
\end{enumerate}
\end{definition}
\noindent
We usually write $s\transc{\act}s'$ instead of $(s,\act,s')\in\transrelg$. 
We also denote $\getacts(\model) = \acts$ and $\getstates(\model) = \states$. 
A \emph{run} in $\lts$ $\model$ is an infinite sequence of states and actions $\run = \loc^0\act^0\loc^1\act^1\ldots$
s.t. $\loc^i\transc{\act_i}\loc^{i+1}$ for all $i\ge 0$.
By $\runs(\model, \loc)$ we denote the set of all the runs starting from state $s\in\states$;
if $\loc$ is the initial state, we simply write $\runs(\model)$.

\subsection{$\lts$ Nets and Synchronisation Topologies}
\label{subsec:automata_networks}

Both the commercial and research model checkers such as \spin{}, \uppaal{} or \imitator{}~\cite{HolzmannSPIN,BehrmannDLHPYH06,AndreFKS12}
typically expect the systems
described in a form of interacting modules. Concurrent transitions via common actions (or \emph{channels})
are one of the most basic synchronisation primitives~\cite{PrinciplesofMC}.

\begin{definition}[Asynchronous Product]
  Let $\model_i = \langle \states_i, \initst_i, \transrelg_i, \acts_i, \labeling_i\rangle$ be
  $\lts$, for $i\in\{1,2\}$.
  The \emph{asynchronous product} of $\model_1$ and $\model_2$ is the $\lts$
  $\model_1\aprod\model_2 = \langle \states_1\times\states_2, (\initst_1,\initst_2), \transrelg, \acts_1\cup\acts_2, \labeling_1\cup\labeling_2 \rangle$ with
  the transition rule defined in the usual way:
\vspace{-1mm}
\begin{align*}
&\infer{%
  (\state_1, \state_2)\transc{\act}{}
  (\state_1', \state_2)
} {%
  \act\in\acts_1\setminus\acts_2 \land \state_1\transc{\act}_{1}\state_1'
}\\
&\infer{%
  (\state_1, \state_2)\transc{\act}
  (\state_1, \state_2')
} {%
  \act\in\acts_2\setminus\acts_1 \land \state_2\transc{\act}_{2}\state_2'
}\\
&\infer{%
  (\state_1, \state_2)\transc{\act}
  (\state_1', \state_2')
} {
  \act\in\acts_1 \cap\acts_2 \land \state_1\transc{\act}_{1}\state_1'
  \land \state_2\transc{\act}_{2}\state_2'
}
\end{align*}
\vspace{-1mm}
\end{definition}
\noindent
The above definition is naturally extended to an arbitrary number of components, where
we sometimes write $\aprod_{i=0}^n \model_i$ instead of $\model_1\aprod\ldots\aprod\model_n$.

The synchronisation topology is an undirected graph that records how components synchronise with one another.

\begin{definition}[Synchronisation Topology]
A \emph{synchronisation topology} (\snt) is a tuple
$\syncn = \langle \tnet, \synce\rangle$,
where $\tnet = \{\model_i\}_{i=1}^n$ is a set of $\lts$ for $i\in\{1,\ldots, n\}$, 
and $\synce\subseteq\tnet\times\tnet$ is s.t.
$(\model_i,\model_j) \in \synce$
iff $i\ne j$ and $\acts_i\cap\acts_j\ne\emptyset$.
\end{definition}

Note that $\synce$ is induced by $\tnet$. Thus, with a slight notational
abuse we sometimes treat $\syncn$ as $\tnet$.
Moreover, we put $\getacts(\syncn) = \bigcup_{i=1}^n\getacts(\model_i)$.

In what follows
we assume that $\syncn$ is a tree with the root $\roottr(\syncn)$. Moreover, for each $\model\in\tnet$
by $\prnt(\model)$ we denote its parent (we assume $\prnt(\roottr(\syncn)) = \emptyset$)
and by $\chld(\model)$ we mean the set of its children. 
By $\upacts(\model)$ (resp., $\downacts(\model)$) we denote the set of actions via which
$\model$ synchronises with its parent (children, resp.).
For each $\act\in\downacts(\model)$ by $\snd(\model,\act)$ we denote the component
$\model'\in\chld(\model)$ s.t. $\act\in\upacts(\model')$. Thus, $\snd(\model,\act)$
is the child of $\model$ that synchronises with $\model$ over $\act$.
If $\model$ is clear from the context, we simply write $\snd(\act)$.
The \emph{local}, unsynchronised actions of $\model$ are defined as
$\locacts(\model) = \getacts(\model)\setminus (\downacts(\model) \cup \upacts(\model))$.
For brevity, whenever we refer to a state or transition of $\syncn$
we mean a state or transition of $\aprod_{i=0}^n \model_i$.
We also extend the notion of runs to synchronisation topologies:
$\runs(\syncn, \loc) = \runs(\aprod_{i=0}^n \model_i, \loc)$ for each
$\loc\in\getstates(\aprod_{i=0}^n \model_i)$.

We are interested in networks whose all components share a similar,
simple structure. Namely, we say that an $\lts$ $\model$ is \emph{live-reset} if every run $\run\in\runs(\model)$
is s.t. executing any action from $\upacts(\model)$ leads to the initial state.
Intuitively, $\model$ can freely synchronise with its children and execute local actions
but resets once synchronising with the parent.
If every component of an $\snt$ $\syncn$ is live-reset then we say that
$\syncn$ is live-reset.

\begin{example}\label{example:simpleltsnet}
  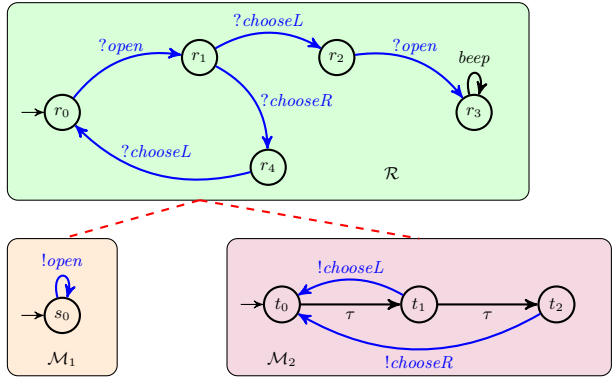
\begin{figure}[t]
    \centering

\begin{tikzpicture}
  [every state/.style={thick,draw,minimum size=6mm}, >=stealth',node distance=1.8cm,->, transform shape, scale = 0.73]

  \def\rootright{6.5}
  \def\roottop{2.0}
  \def\rootbot{-1.6}  
  \def\rootleft{-3.0}  
  \draw[rounded corners, fill=green!15] (\rootleft, \rootbot) rectangle (\rootright, \roottop) {};

  \def\lmodright{-1.0}
  \def\lmodtop{-2.3}
  \def\lmodbot{-4.8}  
  \def\lmodleft{-3.0}  
  \draw[rounded corners, fill=orange!15] (\lmodleft, \lmodtop) rectangle (\lmodright, \lmodbot) {};

  \def\rmodright{8.0}
  \def\rmodtop{-2.3}
  \def\rmodbot{-4.8}  
  \def\rmodleft{1.0}  
  \draw[rounded corners, fill=purple!15] (\rmodleft, \rmodtop) rectangle (\rmodright, \rmodbot) {};

  \draw[-,red,thick,dashed] (0.5,\rootbot) -- (-2.0, \rmodtop);
  \draw[-,red,thick,dashed] (0.5,\rootbot) -- (4.5, \rmodtop);
  
  \node[state,initial] at (\rootleft+1,0) (l0) {$\rtloc_0$};
  \node[state] at (\rootleft+3.5,1) (l1) {$\rtloc_1$};
  \node[state] at (\rootleft+6.0,1) (l2) {$\rtloc_2$};
  \node[state] at (\rootright-1,0) (l3) {$\rtloc_3$};
  \node[state] at ({(\rootleft+\rootright)/2},-1) (l4) {$\rtloc_4$};  

  \draw[sync] (l0) to [bend left] node[above]
       {$\styleSync{\receive\actopen}$} node[below] {} (l1);
  \draw[sync] (l1) to [bend left] node[above]
       {$\styleSync{\receive\actchooseL}$} (l2);
  \draw[sync] (l2) to [bend left] node[above]
       {$\styleSync{\receive\actopen}$} (l3);
  \draw[sync] (l1) to [bend left] node[right]
       {$\styleSync{\receive\actchooseR}$} (l4);
  \draw[sync] (l4) to [bend left] node[above=2mm]
       {$\styleSync{\receive\actchooseL}$} (l0);              
  \draw[nosync] (l3) to [loop above] node[above] {$\actidle$} (l3);
  \node[label] at (\rootright-2.5,\rootbot+0.5) (t2) {$\exroot$};

  \node[state,initial] at (\lmodleft+1,\lmodtop-1.4) (l10) {$\lloc_0$};
  \draw[sync] (l10) to [loop above] node[above]
       {$\styleSync{\send\actopen}$} (l10);                 
  \node[label] at (-2.0,\rmodtop-2.2) (t2) {$\exleft$};       

  \node[state,initial] at (\rmodleft+1,\rmodtop-1.2) (l20) {$\rloc_0$};
  \node[state] at (4.5,\rmodtop-1.2) (l21) {$\rloc_1$};
  \node[state] at (7.0,\rmodtop-1.2) (l22) {$\rloc_2$};    
  \draw[sync] (l21) to [bend right] node[above]
       {$\styleSync{\send\actchooseL}$} (l20);
  \draw[nosync] (l21) to node[below] {$\actsil$} (l22);
  \draw[nosync] (l20) to node[below] {$\actsil$} (l21);       
  \draw[sync] (l22) to [bend left] node[below]
       {$\styleSync{\send\actchooseR}$} (l20);
  \node[label] at (2.0,\rmodtop - 2.2) (t2) {$\exright$};       

\end{tikzpicture}
  \caption{A simple tree synchronisation topology $\syncn_x$.}\label{fig:simpleltsnet}
\end{figure}
\Cref{fig:simpleltsnet} presents a small tree $\snt$ $\syncn_x$ with the root $\exroot$ and two children $\exleft$ and $\exright$.
The auxiliary symbols $?/!$
are syntactic sugar, used to distinguish between $\upacts$ and $\downacts$.
Here, $\upacts(\exroot) = \emptyset$, $\downacts(\exroot) = \{\actopen, \actchooseL, \actchooseR\}$,
and $\locacts(\exroot) = \{\actidle\}$.
Similarly, $\upacts(\exleft) = \{\actopen\}$, $\upacts(\exright) = \{\actchooseL, \actchooseR\}$,
$\downacts(\exleft) = \locacts(\exleft) = \downacts(\exright) = \emptyset$,
and $\locacts(\exright) = \actsil$.
All the components of the model are live-reset.
\end{example}

Let $\syncn = \langle \tnet, \synce \rangle$ be a \snt.
For each $\excomp\in\tnet$ by $\syncn_{\excomp}$ we denote
the \snt induced by the subtree of $\syncn$ rooted in $\excomp$.
Let 
$M\subseteq \tnet$
and $\run*$ be a prefix of some $run\in\runs(\syncn)$ s.t. $\run* = \loc^0\act^0\loc^1\act^1\ldots$
By $\projdown{\run*}(M)$ we denote the projection of $\run*$ to the the product of components in $M$,
\ie the result
of transforming $\run*$ by (1) firstly projecting each $\loc^i$ on the components in $M$;
(2) secondly, removing the actions that do not belong to $M$, together with
their sources.

\lp{Removing sources has to be explained, also why not destinations.}
\mk{TODO - haven't come up with anything yet.}

\begin{example}\label{example:projection}
Consider a sequence:
\begin{align*}
\eta = & \;(\rtloc_0,\lloc_0,\rloc_0)\actsil(\rtloc_0,\lloc_0,\rloc_1)
         \actsil(\rtloc_0,\lloc_0,\rloc_2)\actopen(\rtloc_1,\lloc_0,\rloc_2)\\
         & \actchooseR(\rtloc_4,\lloc_0,\rloc_0)\actsil(\rtloc_4,\lloc_0,\rloc_1)\actchooseL
         (\rtloc_0,\lloc_0,\rloc_0).
\end{align*}
Here, we have $\projdown{\run}(\exroot,\exleft) =
(\rtloc_0,\lloc_0)\actopen(\rtloc_1,\lloc_0)\actchooseR\\
(\rtloc_4,\lloc_0)\actchooseL(\rtloc_0,\lloc_0)$.

\end{example}


\section{Reducing Live-Reset Trees}
\label{sec:live-reset}

In this section we show how to create for a given synchronisation topology $\syncn$ of live-reset components
an $\lts$ that preserves reachability.
The procedure is presented in two steps. Firstly, we show how to build an $\lts$ for two-level trees.
Secondly, we show how to modify the former to deal with trees of arbitrary height in a bottom-up manner.

\subsection{Reduction for Two-level Trees}

Throughout this subsection let $\syncn$
be a live-reset tree \snt with components
$\tnet = \{\exroot, \excomp_1, \ldots, \excomp_n\}$
s.t. $\roottr(\syncn) = \exroot$ and $\chld(\exroot) = \{\excomp_1, \ldots, \excomp_n\}$.
Moreover, let 
$\exroot = \langle \states_\exroot, \initst_\exroot, \acts_\exroot, \transrelg_\exroot, \labeling_\exroot\rangle$
and $\excomp_i = \langle \states_i, \initst_i, \acts_i, \transrelg_i, \labeling_i\rangle$, for $i\in\{1,\ldots,n\}$.
We employ the observations on the nature of synchronisations with live-reset components
in the following definition.

\begin{definition}[Unreduced Sum-of-squares Product]
  Let $\sq^u(\syncn) = \langle \states^u_{sq}, \initstsq,\\ \acts_{sq}, \transrelg_{sq}, \labeling_{sq}\rangle$
  be an $\lts$ s.t.:
  \begin{itemize}
  \item $\states^u_{sq} = \bigcup_{i=1}^n\excomp_i\times\exroot$.
  \item $\initstsq\not\in\states^u_{sq}$ is a fresh initial state.
  \item $\acts_{sq} = \getacts(\syncn) \cup \{\epsilon\}$, where $\epsilon\not\in\getacts(\syncn)$ is a fresh,
    silent action.
  \item The transition relation $\transrelg_{sq}$ is defined as follows:
    \begin{itemize}
    \item $\initstsq\transcr{\epsilon}(\initst_i, \initst_\exroot)$, for all $i\in\{1,\ldots,n\}$;
      intuitively,
      using the new initial state of $\sq(\syncn)$ and $\epsilon$-transitions we can visit the initial
      state of any square product $\excomp_i\times\exroot$.
    \item If $\loc_i\transci{\act}\loc'_i$ and $\act\in\locacts(\excomp_i)$, then
      $(s_i, \loc_\exroot)\transcr{\act}(\loc'_i, \loc_\exroot)$, for each $\loc_\exroot\in\states_\exroot$;
      similarly, if $\loc_\exroot\transcrt{\act}\loc'_\exroot$ and $\act\in\locacts(\exroot)$, then
      $(\loc_i, \loc_\exroot)\transcr{\act}(\loc_i, \loc'_\exroot)$, for each $\loc_i\in\states_i$.
      Thus, the square products are fully asynchronous over local actions.
    \item If $\act\in\upacts(\excomp_i)$, $\loc_i\transci{\act}\initst_i$, and $\loc_\exroot\transcrt{\act}\loc'_\exroot$,
      then $(\loc_i, \loc_\exroot)\transcr{\act}(\initst_j, \loc'_\exroot)$, for all $j\in\{1,\ldots,n\}$. Intuitively,
      after synchronising with $\exroot$, a component $\excomp_i$ will reset and can release control to another
      module.
    \end{itemize}
  \item $\labeling_{sq}(\loc_i, \loc_\exroot) = \labeling_\exroot(\loc_i) \cup \labeling_\exroot(\loc_\exroot)$, for each $(\loc_i, \loc_\exroot)\in\states_{sq}$.
  \end{itemize}
We call $\sq^u(\syncn)$ the Unreduced Sum-of-squares Product of $\syncn$.
\end{definition}

We say that a state $\loc$ of $\syncn$ is \emph{locked} iff there is no run $\run\in\runs(\syncn,\loc)$
s.t. $\run = \loc^0\act^0\loc^1\act^1\ldots$ with $\act^i\in\getacts(\exroot)$, where $\loc^0 = \loc$, for some $i\in\nats$.
\lp{The starting state of $\run$ should be $\loc$ and not $\loc^0$. It is not clear to me if the actions are local or downstream. I understand $\loc$ is in fact a deadlock. If so, then don't name it something else.}
\mk{Fixed the starting state. Is it a deadlock if a child component can still execute actions but the root will never be able to?}
Observe that from a point of view of the root, a locked state is in a full deadlock.
The set of locked states of an $\lts$ can be computed in polynomial time using either a model checker or
conventional graph algorithms.

\begin{definition}[Sum-of-squares Product]
We call the Sum-of-squares Product
$\sq(\syncn)$ of $\syncn$ the result of removing
all the locked states from $\sq^u(\syncn)$
and restricting the relevant transition and labelling functions.
\end{definition}

\begin{figure}[ht]
  \centering  

\begin{tikzpicture}
  [every state/.style={thick,draw,minimum size=6mm}, >=stealth',node distance=1.8cm,->, transform shape, scale = 0.73]

  \def\rootright{1.5}
  \def\roottop{2.3}
  \def\rootbot{-3.8}  
  \def\rootleft{-3.0}  
  \draw[rounded corners, fill=orange!15] (\rootleft, \rootbot) rectangle (\rootright, \roottop) {};

  \def\rmodright{7.2}
  \def\rmodtop{3.5}
  \def\rmodbot{-7.0}  
  \def\rmodleft{2.2}  
  \draw[rounded corners, fill=purple!15] (\rmodleft, \rmodtop+1.1) rectangle (\rmodright, \rmodbot) {};

  \node[state] at (\rootleft+1.5,0) (l0) {$\lloc_0\rtloc_0$};
  \node[state] at (\rootleft+1.5,-1.5) (l2) {$\lloc_0\rtloc_3$};
  \node[state] at (\rootleft+1.5,-3.0) (l3) {$\lloc_0\rtloc_2$};
  \node[state,fill=red!50] at (\rootleft+3.5,0) (l4) {$\lloc_0\rtloc_1$};
  \node[state,fill=red!50] at (\rootleft+3.5,1.5) (l5) {$\lloc_0\rtloc_4$};  

  \draw[sync] (l3) to node[left]
       {$\actopen$} node[below] {} (l2);
  \draw[sync] (l0) to node[above]
       {$\actopen$} (l4);       
  \draw[nosync] (l2) to [loop left] node[above=2mm] {$\actidle$} (l2);

  \node[state] at (\rmodleft+1,\rmodtop-1.2) (l10) {$\rloc_0\rtloc_4$};
  \node[state] at (\rmodleft+1,\rmodtop-2.7) (l11) {$\rloc_2\rtloc_1$};
  \node[state] at (\rmodleft+1,\rmodtop-4.2) (l12) {$\rloc_0\rtloc_1$};
  \node[state] at (\rmodleft+1,\rmodtop-5.7) (l13) {$\rloc_1\rtloc_1$};

  \node[state,fill=red!50] at (\rmodleft+1,\rmodtop-7.5) (l14) {$\rloc_0\rtloc_2$};
  \node[state,fill=red!50] at (\rmodleft+2.5,\rmodtop-7.5) (l15) {$\rloc_1\rtloc_2$};
  \node[state,fill=red!50] at (\rmodleft+4,\rmodtop-7.5) (l16) {$\rloc_2\rtloc_2$};        

  \node[state] at (\rmodleft+1,\rmodtop-8.8) (l17) {$\rloc_0\rtloc_3$};
  \node[state] at (\rmodleft+2.5,\rmodtop-8.8) (l18) {$\rloc_1\rtloc_3$};
  \node[state] at (\rmodleft+4,\rmodtop-8.8) (l19) {$\rloc_2\rtloc_3$};

  \node[state] at (\rmodleft+1,\rmodtop+0.3) (l110) {$\rloc_1\rtloc_4$};

  \node[state,fill=red!50] at (\rmodleft+4,\rmodtop+0.3) (l111) {$\rloc_2\rtloc_4$};    
  \node[state,fill=red!50] at (\rmodleft+4,\rmodtop-1.8) (l112) {$\rloc_0\rtloc_0$};
  \node[state,fill=red!50] at (\rmodleft+4,\rmodtop-3.3) (l113) {$\rloc_1\rtloc_0$};
  \node[state,fill=red!50] at (\rmodleft+4,\rmodtop-4.8) (l114) {$\rloc_2\rtloc_0$};  
  
  \draw[nosync] (l12) to node[left] {$\actsil$} (l13);
  \draw[sync] (l13) to node[right] {$\actchooseL$} (l14);
  \draw[nosync] (l13) to [bend right=40] node[right] {$\actsil$} (l11);    
  \draw[sync] (l11) to node[right] {$\actchooseR$} (l10);
  
  \draw[nosync] (l14) to node[above] {$\actsil$} (l15);
  \draw[nosync] (l15) to node[above] {$\actsil$} (l16);
  \draw[nosync] (l17) to node[above] {$\actsil$} (l18);
  \draw[nosync] (l18) to node[above] {$\actsil$} (l19);
  \draw[nosync] (l17) to [loop below] node[below] {$\actidle$} (l17);
  \draw[nosync] (l18) to [loop below] node[below] {$\actidle$} (l18);  
  \draw[nosync] (l19) to [loop below] node[below] {$\actidle$} (l19);  
  \draw[nosync] (l10) to node[right] {$\actsil$} (l110);

  \draw[sync] (l110) to node[above=3mm] {$\actchooseL$} (l112);

  \draw[nosync] (l112) to node[right] {$\actsil$} (l113);
  \draw[nosync] (l113) to node[right] {$\actsil$} (l114);
  \draw[nosync] (l110) to node[above] {$\actsil$} (l111);  

  \draw[outersync] (l110) to [bend right] node[left=1mm] {$\actchooseL$} (l0);
  \draw[outersync] (l11) to node[above=2mm] {$\actchooseR$} (l5);
  \draw[outersync] (l13) to node[above=2mm] {$\actchooseL$} (l3);
  \draw[outersync] (l3) to [bend right] node[below=1mm] {$\actopen$} (l17);
  \draw[outersync] (l0) to [bend right] node[below=1mm] {$\actopen$} (l12);            

  \node[state,fill=green!15] at (\rootleft+2.4, 4.2) (fin) {$\initst_{sq}$};
  \draw[sync,style=dotted] (fin) to node[above] {$\epsilon$} (l112);
  \draw[sync,style=dotted] (fin) to [bend right] node[left] {$\epsilon$} (l0);  
  
\end{tikzpicture}
  \caption{The unreduced sum-of-squares product of a simple tree \snt $\syncn_x$.}\label{fig:simpleltsnetsq}
\end{figure}
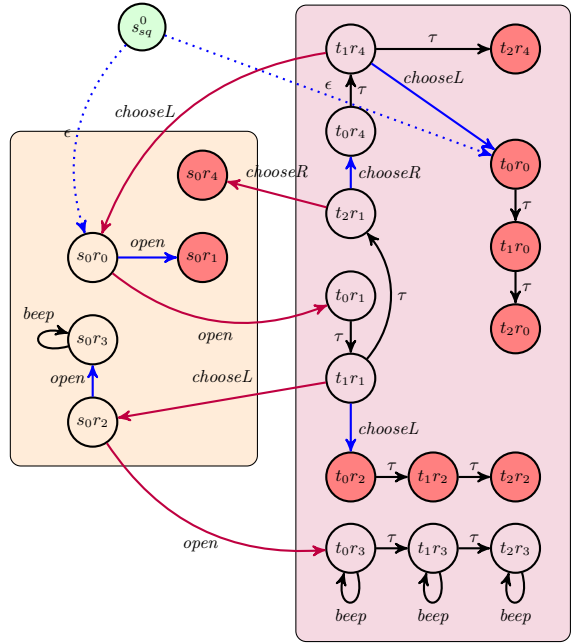

\begin{figure}[!t]
  \centering  

\begin{tikzpicture}
  [every state/.style={thick,draw,minimum size=6mm}, >=stealth',node distance=1.8cm,->, transform shape, scale = 0.73]

  \def\rootright{1.5}
  \def\roottop{2.3}
  \def\rootbot{-3.8}  

  \def\rmodright{7.2}
  \def\rmodbot{-7.0}  

  \node[state] at (-1.5,0.8) (l0) {$\lloc_0\rtloc_0$};
  \node[state] at (4.2,-3.0) (l2) {$\lloc_0\rtloc_3$};
  \node[state] at (2.2,-3.0) (l3) {$\lloc_0\rtloc_2$};

  \draw[sync] (l3) to node[below=2mm]
       {$\actopen$} (l2);
  \draw[nosync] (l2) to [loop below] node[below] {$\actidle$} (l2);

  \node[state] at (2.2,1.7) (l10) {$\rloc_0\rtloc_4$};
  \node[state] at (4.2,0.8) (l11) {$\rloc_2\rtloc_1$};
  \node[state] at (0.5,-0.7) (l12) {$\rloc_0\rtloc_1$};
  \node[state] at (2.2,-0.7) (l13) {$\rloc_1\rtloc_1$};

  \node[state] at (0.1,-3.0) (l17) {$\rloc_0\rtloc_3$};
  \node[state] at (-1.5,-3.0) (l18) {$\rloc_1\rtloc_3$};
  \node[state] at (-3,-3.0) (l19) {$\rloc_2\rtloc_3$};

  \node[state] at (0.5,1.7) (l110) {$\rloc_1\rtloc_4$};

  \draw[nosync] (l12) to node[below] {$\actsil$} (l13);
  \draw[nosync] (l13) to [bend right] node[below=1mm] {$\actsil$} (l11);    
  \draw[sync] (l11) to [bend right] node[above=2mm] {$\actchooseR$} (l10);
  
  \draw[nosync] (l17) to node[above] {$\actsil$} (l18);
  \draw[nosync] (l18) to node[above] {$\actsil$} (l19);
  \draw[nosync] (l17) to [loop below] node[below] {$\actidle$} (l17);
  \draw[nosync] (l18) to [loop below] node[below] {$\actidle$} (l18);  
  \draw[nosync] (l19) to [loop below] node[below] {$\actidle$} (l19);  
  \draw[nosync] (l10) to node[above] {$\actsil$} (l110);

  \draw[outersync] (l110) to [bend right] node[above=2mm] {$\actchooseL$} (l0);
  \draw[outersync] (l13) to node[right] {$\actchooseL$} (l3);
  \draw[outersync] (l3) to node[below=2mm] {$\actopen$} (l17);
  \draw[outersync] (l0) to [bend right] node[below=2mm] {$\actopen$} (l12);            

  \node[state,fill=green!15] at (-3,-1) (fin) {$\initst_{sq}$};
  \draw[sync,style=dotted] (fin) to [bend right] node[below] {$\epsilon$} (l0);
  
\end{tikzpicture}
  \caption{The sum-of-squares product of a simple tree \snt $\syncn_x$.}\label{fig:simpleltsnetsqr}
\end{figure}
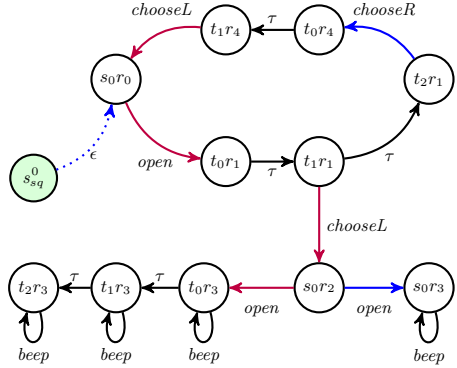

\begin{example}\label{example:simpleltsnetsq}
\cref{fig:simpleltsnetsq} presents the unreduced sum-of-squares product $\sq^u(\syncn_x)$ for
the small tree \snt from Example~\ref{example:simpleltsnet}.
The locked states are coloured red. 
Fig.~\ref{fig:simpleltsnetsqr} displays the sum-of-squares product
$\sq(\syncn_x)$ 
of the topology. Note the similarity of the model to the root of Fig.~\ref{fig:simpleltsnet}
that reveals that the children do not restrict the root's
freedom.

\end{example}

As shown in the above example, it is possible that the size of the state space of an (unreduced) sum-of-squares product
of a live tree \snt $\syncn$ is equal to or greater than the size of the state space of $\syncn$.
On the other hand, the size of a representation of a state will be smaller in (unreduced) sum-of-squares product,
as it records only local states of at most two components of the network.
However, in less degenerate cases than our toy model we can expect significant reductions.
In particular, if a two-level tree \snt contains $n$ components,
where the statespace of each is of size $m$, then the size of its asynchronous product
can reach $m^n$. In contrast, the size of the (unreduced) sum-of-squares product
of such topology is at most $(n-1)\cdot m^2$.
The structure of the sum-of-squares is similar to the structure of the root. This construction
preserves reachability, but not the $EG$ modality of $\CTL$, as shown in
Proposition~\ref{theorem:sqnpreservG}.

\begin{figure}[ht]
  \centering

\begin{tikzpicture}
  [every state/.style={thick,draw,minimum size=6mm}, >=stealth',node distance=1.8cm,->, transform shape, scale = 0.73]

  \def\rootleft{-2.0}
  \def\roottop{1.5}
  \def\rootright{\rootleft+7}
  \def\rootbot{-1.0}  
  \draw[rounded corners, fill=green!15] (\rootleft, \rootbot) rectangle (\rootright, \roottop) {};

  \def\lmodleft{-3.0}    
  \def\lmodright{\lmodleft + 3.8}
  \def\lmodtop{-2.3}
  \def\lmodbot{-4.8}  
  \draw[rounded corners, fill=orange!15] (\lmodleft, \lmodtop) rectangle (\lmodright, \lmodbot) {};

  \def\rmodleft{2.0}    
  \def\rmodright{\rmodleft + 3.8}
  \def\rmodtop{-2.3}
  \def\rmodbot{-4.8}  

  \draw[rounded corners, fill=purple!15] (\rmodleft, \rmodtop) rectangle (\rmodright, \rmodbot) {};

  \draw[-,red,thick,dashed] ({\rootleft+\rootright)/2},\rootbot) -- (-2.0, \rmodtop);
  \draw[-,red,thick,dashed] ({\rootleft+\rootright)/2},\rootbot) -- (4.5, \rmodtop);
  
  \node[state,initial] at (\rootleft+1,0) (l0) {$\rtloc_0$};
  \node[state] at (\rootleft+3.5,0) (l1) {$\rtloc_1$};
  \node[state,label=below:{$p$}] at (\rootleft+6.0,0) (l2) {$\rtloc_2$};

  \draw[sync] (l0) to [bend left] node[above]
       {$\styleSync{\receive\actchooseR}$} node[below] {} (l1);
  \draw[sync] (l1) to [bend left] node[above]
       {$\styleSync{\receive\actchooseL}$} (l2);
  \draw[nosync] (l2) to [loop above] node[above]
       {$\actidle$} (l2);
  \node[label] at ({(\rootleft+\rootright)/2},\rootbot+0.3) (t2) {$\exrootn$};

  \node[state,initial] at (\lmodleft+1,{(\rmodtop+\rmodbot)/2}) (l10) {$\lloc_0$};
  \node[state,label=below:{$p$}] at (\lmodleft+3,{(\lmodtop+\lmodbot)/2}) (l11) {$\lloc_1$};
  \draw[sync] (l11) to [bend right] node[above]
       {$\styleSync{\send\actchooseL}$} (l10);
  \draw[nosync] (l10) to node[below] {$\actsil$} (l11);       
  \node[label] at ({(\lmodleft+\lmodright)/2},\rmodtop - 2.2) (t2) {${\exleftn}$};  

  \node[state,initial,label=below:{$p$}] at (\rmodleft+1,{(\rmodtop+\rmodbot)/2}) (l20) {$\rloc_0$};
  \node[state] at (\rmodleft+3,{(\rmodtop+\rmodbot)/2}) (l21) {$\rloc_1$};
  \draw[sync] (l21) to [bend right] node[above]
       {$\styleSync{\send\actchooseR}$} (l20);
  \draw[nosync] (l20) to node[below] {$\actsil$} (l21);       
  \node[label] at ({(\rmodleft+\rmodright)/2},\rmodtop - 2.2) (t2) {${\exrightn}$};       

\end{tikzpicture}
  \caption{Sum-of-squares does not preserve EG.}\label{fig:noeg}
\end{figure}
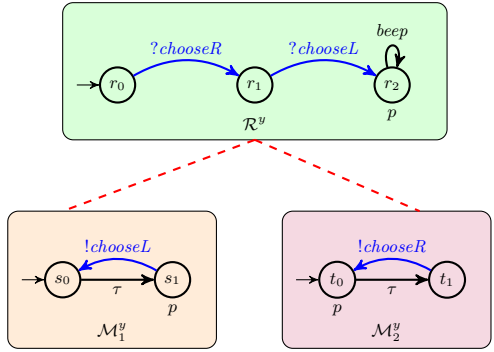

\begin{theorem}[Sum-of-squares Preserves Reachability]\label{theorem:sqpreservF}
Let $\syncn$ be a live two-level tree \snt.
For each $p\in\props$
$\syncn\models EFp$ iff
$\sq(\syncn)\models EFp$.
\end{theorem}
\begin{proof}
  Recall that we assume $\tnet = \{\exroot, \excomp_1, \ldots, \excomp_n\}$
  with root $\exroot$ and children $\{\excomp_i\}_{i=1}^n$.
  Let $\syncn\models EFp$ and $\run = \loc^0\act^0\loc^1\act^1\ldots$
  be a run of $\syncn$ s.t. $p\in\labeling(\loc_i)$ for some $i\in\nats$.
  Now, $\run$ can be represented as $\run = \alpha_1 F_1 \alpha_2 F_2 \ldots$,
  where for each $i\in\nats$ there exist $j,k\in\nats$ such that
  $\alpha_i = \loc^j\act^j\ldots\loc^k\act^k\loc^{k+1}$
  and
  $\act^j,\ldots,\act^k\in\locacts(\exroot)\cup\bigcup_{i=1}^n\locacts(\excomp_i)$
  and $F_i\in\downacts(\exroot)$.
  Actions are never synchronised between children, thus it can be proven by
  induction on the length of the run that the actions in 
  $\run$ can be reordered to obtain a run $\run'\in\runs(\syncn, \loc^0)$
  that can be represented as
  $\run' = \alpha'_1 F_1 \alpha'_2 F_2 \ldots$, such that:
  \begin{enumerate}
  \item For any $i\in\nats$ there exist $j,k\in\nats$ such that
  $\alpha'_i = \loc'^j\act'^j\ldots\loc'^k\act'^k\loc'^{k+1}$
  and $\act^j,\ldots,\act^k\in\locacts(\exroot)\cup\locacts(\snd(F_i))$.
  \item  For each $i\in\nats$ and $1\le j\le n$ we have
  $\projdown{\alpha_i}(\exroot, \snd(F_i)) = \projdown{\alpha'_i}(\exroot, \snd(F_i))$.
  \item For each $\loc'^j$ in $\alpha_i$, if $0$ is the coordinate of root and
    $k$ is the coordinate of $\snd(F_i)$, then $\loc'^j = (\loc_{0},\initst_{1},\ldots,\initst_{k-1},\loc_{k},\initst_{k+1},\ldots)$
    for some $\loc_0\in\getstates(\exroot)$, $\loc_k\in\getstates(\snd(F_i))$.
  \end{enumerate}
  Intuitively, $\run'$ is built from $\run$ in such a way that firstly
  only the root and the component that synchronises with the root over $F_1$
  are allowed to execute their local actions while all the other components
  stay in their initial states; then $F_1$ is fired;
  and then this scheme is repeated for $F_2, F_3$, etc.
  We can now project $\run'$ on spaces of squares of the root and components
  active in a given interval, to obtain
  $\run'' = \projdown{\alpha'_i}(\exroot, \snd(F_1)) F_1 \projdown{\alpha'_i}(\exroot, \snd(F_2)) F_2 \ldots$
  As $\run''\in\sq(\syncn)$ and it can be observed that $\run''$ visits each local state
  that appears along $\run$, this part of the proof is concluded.

  Let $\sq(\syncn)\models EFp$ and $\run\in\runs(\sq(\syncn))$ visit a state labelled with $p$.
  Now, it suffices to replace in $\run$ each state $(\loc_k,\loc_0)$ that belongs
  to the square $\excomp_k\times\exroot$ with the global state
  $(\loc_{0},\initst_{1},\ldots,\initst_{k-1},\loc_{k},\initst_{k+1},\ldots)$  
  of $\syncn$. The result of this substitution is a run of $\syncn$ that visits $p$.
\end{proof}

\begin{proposition}[Sum-of-squares Does Not Preserve $EG$]\label{theorem:sqnpreservG}
There exists a live two-level tree \snt $\syncn$ s.t.
for some $p\in\props$
$\syncn\models EGp$ and
$\sq(\syncn)\not\models EGp$. 
\end{proposition}
\begin{proof}

Consider the tree \snt $\syncn_y$ in Fig.~\ref{fig:noeg}.
Here, we have $\syncn_y\models EGp$, but each path $\run$ along which $p$ holds
globally, starts with $\exleftn$ executing $\actsil$ followed by
$\exrightn$ executing $\actsil$ and, consecutively, $\actchooseR$.
Thus, it is not possible to partition $\run$ into intervals where
one child executes local actions until synchronisation with the
root and possible release of control to another child.
Hence, $\sq(\syncn_y)\not\models EGp$.  
\end{proof}

\lp{I agree. The reason is that the sum of squares does not mimic the interleaving of local actions. Couldn't a reordering of local actions of children suffice?}
\mk{I don't understand...}

\subsection{Adaptation for Any Tree Height}

It is rather straightforward to adapt the sum-of-squares of a subtree 
to allow for synchronisation of the root with its parent.
By $\cmpl(\sq(\syncn))$ we denote the result of replacing in $\sq(\syncn)$
every transition $(\loc,\act,\loc')$, where $\act\in\upacts(\exroot)$ with
$(\loc,\act,\initstsq)$.
Note that $\cmpl(\sq(\syncn))$ is a live-reset tree \snt.

\begin{example}\label{example:syncanylevel}
To obtain $\cmpl(\sq(\syncn_x))$ for the sum-of-squares product from 
Fig.~\ref{fig:simpleltsnetsqr} move the targets of the looped
$\actidle$ transitions to $\initstsq$.
\end{example}

We are now ready to provide the algorithm for reducing any
live tree \snt to a single component
while preserving reachability.

\begin{algorithm}[!htp]
\caption{$\reduce(\syncn)$}\label{algo:full_reduction}
\textbf{Input: live-reset tree sync. topology $\syncn$}\\
\textbf{Output: }$\lts$ $\model$ s.t. $\syncn\models EFp$ iff
$\model_{\syncn}\models EFp$.\\
\vspace{-3mm}
   \begin{algorithmic}[1]

   \IF{$|\scomp(\syncn)| = 1$}
     \RETURN{$\syncn$ \text{(* $\syncn$ is a leaf *)}}
   \ENDIF

   \STATE{{\bf let }$\reducedchldn := \emptyset$}

   \FOR{$\chldalg\in\chld(\roottr(\syncn))$}
    \STATE $\reducedchldn.append(\reduce(\syncn_{\chldalg}))$
   \ENDFOR

   \STATE{{\bf let }$\syncn' := \{\roottr(\syncn)\} \cup \reducedchldn$}

   \RETURN{$\cmpl(\sq(\syncn'))$}

   \end{algorithmic}
\end{algorithm}

Algorithm~\ref{algo:full_reduction} applies the two-level reduction
$\cmpl(\sq(\cdot))$ to all the nodes of the $\snt$, in a bottom-up manner.
Its soundness and correctness is expressed by the following theorem.

\begin{theorem}[$\reduce(\syncn)$ Preserves Reachability]\label{theorem:reducepreservF}
Let $\syncn$ be a live tree \snt.
For each $p\in\props$
$\syncn\models EFp$ iff
$\reduce(\syncn)\models EFp$.
\end{theorem}
\begin{proof}(Sketch)
  The proof follows via induction on the height of the tree $\syncn$.
  As we have Theorem~\ref{theorem:sqpreservF}, it suffices to prove that $\cmpl(\sq(\syncn))$
  preserves reachability for any two-level live $\snt$ $\syncn$.
  This, however, can be done in a way very similar to the proof of Theorem~\ref{theorem:sqpreservF}
  and is omitted.
\end{proof}


\section{Conclusion}
\label{sec:conclusion}

In this paper we have outlined how to simplify large tree networks of automata that reset
after synchronising with their parents. It is shown that the reduction preserves a
certain form of reachability, but it does not preserve safety.
While the procedure is quite fast and effective, it has several limitations.
Firstly, it preserves reachability of labelings, but not their conjunctions;
namely, it is not guaranteed that $\reduce(\syncn)\models EF(p\land q)$
iff $\syncn\models EF(p\land q)$.
Secondly, we would like to relax the assumption that all the components are live-reset automata.
It is not difficult to see how to adapt the original construction to the general case.
To this end it suffices to extend the sum-of-squares product with an explicit model of the memory
of last synchronisations.
Interacting modules can then use this memory to register the \emph{return states}, \ie{} the locations
entered after synchronising action. Thus, a synchronising step between a root and one of its children
would become a process consisting of the following four steps: 
(1) perform joint synchronising transition; (2) record the target locations; (3) make
a non-deterministic selection of a child and read from the memory its return state; (4) continue
the execution of the pair of the root and the new child.
This construction, however, can hinder the expected reduction due to the size of the memory component.
Finally, it is possible that the assumption of tree-like communication between the components
is too strong for any real-life applications. Thus it should be investigated if the
proposed procedures can be easily extended to other topologies.

We plan to address these limitations in future work.



  


\bibliographystyle{IEEEtran}
\bibliography{refs}

\begin{thebibliography}{1}
\providecommand{\url}[1]{#1}
\csname url@samestyle\endcsname
\providecommand{\newblock}{\relax}
\providecommand{\bibinfo}[2]{#2}
\providecommand{\BIBentrySTDinterwordspacing}{\spaceskip=0pt\relax}
\providecommand{\BIBentryALTinterwordstretchfactor}{4}
\providecommand{\BIBentryALTinterwordspacing}{\spaceskip=\fontdimen2\font plus
\BIBentryALTinterwordstretchfactor\fontdimen3\font minus
  \fontdimen4\font\relax}
\providecommand{\BIBforeignlanguage}[2]{{%
\expandafter\ifx\csname l@#1\endcsname\relax
\typeout{** WARNING: IEEEtran.bst: No hyphenation pattern has been}%
\typeout{** loaded for the language `#1'. Using the pattern for}%
\typeout{** the default language instead.}%
\else
\language=\csname l@#1\endcsname
\fi
#2}}
\providecommand{\BIBdecl}{\relax}
\BIBdecl

\bibitem{LTR}
``{LTR},'' \url{https://github.com/MichalKnapik/automata-net-reduction-tool},
  2020.

\bibitem{HolzmannSPIN}
G.~J. Holzmann, \emph{The {SPIN} Model Checker - primer and reference
  manual}.\hskip 1em plus 0.5em minus 0.4em\relax Addison-Wesley, 2004.

\bibitem{BehrmannDLHPYH06}
G.~Behrmann, A.~David, K.~G. Larsen, J.~H{\aa}kansson, P.~Pettersson, W.~Yi,
  and M.~Hendriks, ``{UPPAAL} 4.0,'' in \emph{Third International Conference on
  the Quantitative Evaluation of Systems {(QEST} 2006), 11-14 September 2006,
  Riverside, California, {USA}}.\hskip 1em plus 0.5em minus 0.4em\relax {IEEE}
  Computer Society, 2006, pp. 125--126.

\bibitem{AndreFKS12}
{\'{E}}.~Andr{\'{e}}, L.~Fribourg, U.~K{\"{u}}hne, and R.~Soulat, ``{IMITATOR}
  2.5: {A} tool for analyzing robustness in scheduling problems,'' in
  \emph{{FM}~2012}, ser. {LNCS}, vol. 7436.\hskip 1em plus 0.5em minus
  0.4em\relax Springer, 2012, pp. 33--36.

\bibitem{PrinciplesofMC}
C.~Baier and J.~Katoen, \emph{Principles of model checking}.\hskip 1em plus
  0.5em minus 0.4em\relax {MIT} Press, 2008.

\end{thebibliography}
\end{document}